\documentclass{article}
\usepackage[dvips]{graphicx}
\usepackage{graphicx,epsfig}
\usepackage{amsmath,amsfonts,amssymb,amsthm,amsbsy,mathtools}
\usepackage{bbold}

\newtheorem{lem}{Lemma}
\newtheorem{defi}{Definition}
\newtheorem*{col}{Corollary}
\newtheorem{rem}{Remark}

\newcommand{\ket}[1]{\mathop{\left|#1\right>}\nolimits}            
\newcommand{\ketf}[1]{\mathop{\lfloor#1\rangle}\nolimits}       
\newcommand{\bra}[1]{\mathop{\left<#1\,\right|}\nolimits}         
\newcommand{\brk}[2]{\langle #1 | #2 \rangle}
\newcommand{\bkf}[2]{\langle #1 \rfloor\hspace{-4pt}\lfloor\, #2 \rangle}          
\newcommand{\kbf}[2]{\lfloor #1\rangle \hspace{-4pt}\,\langle #2 \rfloor}          
\newcommand{\kbr}[2]{| #1\rangle\!\langle #2 |}
\newcommand{\atanh}[1]{\mathop{{\mathrm{atanh{\mit{\ #1}}}}}}  

\newcommand{\nn}{\nonumber}
\def\p{\prime}
\newcommand{\dg}{\dagger}

\def\openone{\mathbb{1}}
\def\C{\mathbb{C}}

\def\R{\mathbb{R}}

\def\CC{\mathcal{C}}

\def\H{\mathcal{H}}

\def\F{\mathcal{F}}

\def\a{\alpha}
\def\b{\beta}

\def\d{\delta}
\def\e{\epsilon}
\def\ve{\varepsilon}
\def\vr{\varrho}
\def\vp{\varphi}

\def\s{\sigma}
\def\vs{\varsigma}
\def\la{\lambda}
\def\La{\Lambda}
\def\vt{\vartheta}
\def\t{\theta}

\def\ketp{\ket{{p},\s}}
\def\ketq{\ket{{q},\s}}

\def\crop{a^\dg(p,\s)}
\def\cropp{a(p',\s')}

\begin{document}

\title{Relativistically invariant photonic wave packets}

\author{Kamil Br\'adler\thanks{kbradler@cs.mcgill.ca}\\\\School of Computer Science\\McGill University\\Montreal\\Quebec, Canada}

\maketitle

\begin{abstract}
We present a photonic wave packet construction  which is immune against the decoherence effects induced by the action of the Lorentz group. The amplitudes of a pure quantum state representing the wave packet remain invariant irrespective of the reference frame into which the wave packet has been transformed. Transmitted information is encoded in the helicity degrees of  freedom of two correlated momentum modes. The helicity encoding is considered to be particularly suitable for free-space communication.  The integral part of the story is information retrieval on the receiver's side. We employed probably the simplest possible helicity (polarization) projection measurement originally studied by Peres and Terno. Remarkably, the same conditions  ensuring the invariance of the wave packet also guarantee perfect distinguishability in the process of measuring the helicity.
\end{abstract}

\section{Introduction}

The helicity density matrix of a generic momentum-helicity wave packet under the action of the Lorentz group is ill-defined. The underlying reason is the explicit helicity dependence on the momentum eigenstates when transformed from one reference frame into another. Consequently, tracing over momentum is an act of violence against the rules of quantum mechanics because one intends to sum helicity density matrices from different Hilbert spaces. It makes no sense to talk about transformation properties of the helicity density matrices and there is no group representation the helicity matrix can transform according to. The situation is further worsened by attempts to measure such a state to recover the encoded information. This problem was first systematically studied in~\cite{pete03} but some explicit comments were already made in~\cite{arvmuk} and probably even before. A partial remedy was found in~\cite{pete03} too where the authors first  realized the importance of a specific helicity measurement and then introduced the concept of a three-dimensional helicity density matrix to deal with the problems mentioned above. Regarding the helicity measurement, it is a simple projective measurement onto the plane perpendicular to the direction of propagation of the wave packet and due to its simplicity it will be used in this work as well. We will refer to this kind of measurement as the Peres-Terno (measurement) scheme. We will, however, abandon the concept of three-dimensional helicity density matrix. The additional longitudinal degree of freedom introduced in~\cite{pete03} for the sake of having at least definite transformation properties under ordinary rotations will be of no use here.  Moreover, the introduction of a longitudinal degree of freedom sounds like there is something forbidden going on but it is not the case. This expression, however, simply refers to the third degree of freedom of a specific POVM measurement. One could in principle have POVMs with arbitrarily many outcomes to be able to reconstruct the wave packet to an arbitrary precision. The whole problem of the helicity density matrix has been further clarified and generalized to other (non-relativistic) situations in~\cite{AW1} and~\cite{linter}. Problems of density matrix construction were considered in other scenarios, too~\cite{fry}. For the purpose of this article we stick to the relativistic context.

There exists a class of wave packets where the definition of the helicity density matrices causes no troubles. These are linear polarization wave packets in which case, loosely speaking, all the momentum three-vectors point in the same direction (the direction of propagation). This class of wave packets has very simple transformation properties and allows for introduction of the invariant helicity density matrix~\cite{czachor,caban,alsing}. That is, tracing over momentum is allowed and by invariance we mean that the coefficients of the resulting density matrix are independent of the reference frame. States transforming covariantly under a general action of the Lorentz group are a subject of interest. The classification of two-photon states has been provided in Ref.~\cite{caban_cov}. Let us emphasize that contrary to our paper only states with sharp momenta were considered so the physical picture considerably differs.

The aim of this paper is twofold. In the first part we construct localized wave packets where the information is encoded in the helicity degrees of freedom (Sec.~\ref{sec:invariantpacket}). The effect of an arbitrary Lorentz transformation is studied and the conditions for invariant transformation of these wave packets are found. The helicity density matrix derived from these states is well defined and is invariant under the general action of the Lorentz group. Finally, the effect of the Peres-Terno helicity measurement on the receiver's side is investigated concluding the perfect discriminability of two initially orthogonal states in contrast to the original solution. Appendix~A contains some further technical details.

The second part is solely aimed at non-experts on the problem of the photonic wave packets construction and their relativistic transformations. It consists of Appendices B and C where we first briefly review Wigner's phenomenal contribution to the understanding of the role which is played by the Poincar\'e group in quantum field theory in Minkowski spacetime~\cite{wi39}. 
Appendix~C contains a discussion of a rather complicated issue of measurement for generic photonic wave packets and the solution of this problem found in~\cite{pete03}.

\section{Relativistically invariant photonic wave packets}\label{sec:invariantpacket}

\subsection{Photonic wave packets}

We adopt the following notation throughout the article. The components of a general momentum four-vector (briefly four-momentum) $p$ is written  as $p^\mu$. We recall that indices from the Greek alphabet are reserved for labeling all four components of the four-momentum. The four-momenta $p^\mu$ living in the four-momentum space can be of null ($p_\mu p^\mu=0$), timelike ($p_\mu p^\mu>0$) or spacelike ($p_\mu p^\mu<0$) character (omitting the zero vector). For the purpose of this article we are interested in the first class. Hence we set $c=1$ and normalize energy such that $p^0=1$ so we get $p_ip^i=1$ ($i\in\{1,2,3\}=\{x,y,z\}$) for all null four-momenta from the future light cone. For further details about the Poincar\'e group and how it acts the reader can take a look at Appendix~B. Unless explicitly mentioned the manifold we will consider in this paper is the four-momentum space and not the usual Minkowski spacetime (one temporal and three spatial coordinates). The difference is again sketched in Appendix~B. For a certainly better exposition, one might also consult extensive literature on the subject of the representations of the Poincar\'e group, for example,~in~\cite{Wein,Tung,halpern,sexl}.

A generic single-particle pure photonic wave packet can be written in the form
\begin{equation}\label{eq:genwavepacket}
    \ket{\psi}=\int{\rm d}\mu(p)f(p)\big(\a_{p}\ket{p,+}+\b_{p}\ket{p,-}\big),
\end{equation}
where $\a_ p,\b_ p\in\C,|\a_{ p}|^2+|\b_{ p}|^2=1,\int{\rm d}\mu( p)|f( p)|^2=1,{\rm d}\mu( p)={1/(2(2\pi)^3)}{\rm d}^3p$ is the Lorentz invariant measure and $f( p)$ is an envelope (scalar) function therefore staying invariant under the Lorentz transformation~\cite{Tung}. The helicity eigenstates span a two-dimensional complex Hilbert space and these degrees of freedom are perfectly suited for transmission of quantum or classical information in free space.  We keep the zeroth component of all four-momentum vectors equal to one (or equivalently that $p_ip^i=1$ holds) so we restrict ourselves to monochromatic wave packets. We did not make our considerations less general by restricting to monochromatic wave packets. The measurement introduced later will  simply be color-blind, that is, it will act trivially on the color subspace $\CC_ p$ for every direction $p$.

Furthermore, we naturally assume that $\a,\b$ are identical for all momentum directions so we may drop the subscript. The preparation of the wave packet is assumed to be under our control. Moreover, it would be a strange way of encoding quantum information into the wave packet if the quantum state differed for different momentum four-vectors already at the preparation stage. However, the momentum directions (the envelope function $f(p)$) continue to be distributed arbitrarily. We thus rewrite Eq.~(\ref{eq:genwavepacket}) as
\begin{equation}\label{eq:genwavepacket_rewritten}
    \ket{\psi}=\a\int{\rm d}\mu(p)f(p)\ket{p,+}+\b\int{\rm d}\mu(p)f(p)\ket{p,-}.
\end{equation}

Let us see what happens if we first just rotate the wave packet
\begin{equation}\label{eq:wavpackrotated}
    U\big(R(\vt,\vp)\big)\ket{\psi}
    =\a\int{\rm d}\mu(q)f(q)\exp{[i\t_{p,q}]}\ket{q,+}+\b\int{\rm d}\mu(q)f(q)\exp{[-i\t_{p,q}]}\ket{q,-},
\end{equation}
where $R(\vt,\vp)$ determines the spatial rotation of the coordinate system (or of the wave packet if we adopt the active point of view) in the direction given by $\vt$ and $\vp$ and from the scalar character of the envelope function it is understood that $f(q)\equiv f(p)$. The transformation in Eq.~(\ref{eq:wavpackrotated}) is indeed nothing else than a $U(1)$ transformation and the parameter is the Wigner angle from Eq.~(\ref{eq:unitarylorentz}). Also, from Eq.~(\ref{eq:unitarylorentz}) follows that the Wigner angle is explicitly dependent on $p$ and so every ket under the above integral acquires a different phase. As pointed out by the authors of~\cite{pete03} when we trace over the non-interesting degrees of freedom (the momentum) the resulting two-dimensional helicity density matrix does not transform according to a representation of the $SU(2)$ group. It does not actually have transformation properties at all. We can get an insight into why tracing over the momentum is ill-defined. Take just two different directions $p_1, p_2$ in Eq.~(\ref{eq:genwavepacket_rewritten}). The rotation of the wave packet  given by $R(\vt,\vp)$ results in an unequal phase change and thus the coefficients originally equal to each other will generally become different. So by a suitable rotation two vectors with the same components can be made orthogonal and viceversa. Note, however, that even before the rotation the operation of tracing over the momentum is not valid despite of bringing up a reasonable density matrix.

\subsection{Assembling of invariant photonic wave packets}

We will take a slightly different path towards a reasonable wave packet density matrix. While we stick to the helicity measurement studied by Peres and Terno~\cite{pete03} we change the building block from which a wave packet is constructed. The choice of the Peres-Terno measurement scheme is well motivated. This helicity measurement is probably the simplest one possible. We recall the whole issue for readers unfamiliar with the problem in Appendix~C. The fixed projection measurement breaks the rotational symmetry. This is the reason why we cannot hope in a construction of a covariant density matrix.

Nevertheless, we will be able to construct an invariant density matrix. That is a matrix whose components will stay constant independently on the receivers reference frame. How do we achieve that? Contrary to~\cite{pete03} we use helicity-momentum entangled two-mode states as the construction blocks. We let the helicity-momentum entangled two-mode states propagate  in the $z$ direction in a fixed coordinate system. We first find the conditions for the entangled states under which they do not acquire any phase in a different (spatially rotated) reference frame. Having these conditions imposed, we investigate the effect of the helicity measurement on the entangled states and also see the effect of a boost in the $z$ direction. Based on the results, we shall construct wave packets which can perfectly transfer quantum or classical information even in the relativistic regime.

\subsection{The effect of spatial rotations}



For the purpose of the helicity wave packet construction introduced later we will now show that for the Bell states (Eq.~(\ref{eq:Bellstates})) with specifically correlated momenta the effect of a spatial rotation is harmless.

Following Appendix~B, let us introduce the helicity eigenstates of the standard momentum $k$
\begin{equation}\label{eq:projpol}
    \ket{+}_ k=\begin{pmatrix}
                               1 \\
                               0 \\
                                \,0\, \\
                             \end{pmatrix},
    \quad
    \ket{-}_ k=\begin{pmatrix}
                               0 \\
                               1 \\
                               \,0\, \\
                             \end{pmatrix}.
\end{equation}
The axial component can be disregard as a result of the application of the Coulomb gauge which explicitly sets the component to zero. The helicity eigenstates can be transformed into the linear polarization eigenstates corresponding to the standard momentum $k$ by means of the following unitary operation
\begin{equation}\label{eq:Pmatrix}
    S={1\over\sqrt{2}}\begin{pmatrix}
         1 & -i & 0 \\
         1 & i & 0 \\
         0 & 0 & \sqrt{2} \\
       \end{pmatrix}.
\end{equation}
We will need this operation to transform the rotation operator from the linear polarization basis $R(\vt,\vp)$~\cite{pete03} to the helicity basis $\tilde R(\vt,\vp)$. Therefore, to determine the helicity vector for an arbitrary direction parametrized by two angles $\vt,\vp$ we transform the helicity eigenstates as $\ket{\pm}_ p=\tilde R(\vt,\vp)\ket{\pm}_ k=S\,R(\vt,\vp)\,S^{-1}\ket{\pm}_ k$. Then
\begin{equation}\label{eq:helrotation}
    \tilde R(\vt,\vp)={1\over2}\left(\begin{matrix}
                           (\cos\vt+1)\exp{(-i\vp)} & (\cos\vt-1)\exp{(-i\vp)} & \sqrt{2}\sin{\vt}\exp{(-i\vp)}\\
                            (\cos\vt-1)\exp{(i\vp)} & (\cos\vt+1)\exp{(i\vp)} & \sqrt{2}\sin{\vt}\exp{(i\vp)}\\
                            -\sqrt{2}\sin\vt  & -\sqrt{2}\sin\vt    & 2\cos\vt
                         \end{matrix}\right)
\end{equation}
and so throughout the article we work strictly in the helicity basis~\footnote{Note that we get the helicity eigenstates in the linear polarization basis by the action of $S^{-1}\ket{\pm}_ k$. Before proceeding, let us mention that there exist many conventions for the matrix $S$. One can, for example, often find
$$
    S^\p={1\over\sqrt{2}}\begin{pmatrix}
         -1 & 1 & 0 \\
         -i & -i & 0 \\
         0 & 0 & \sqrt{2} \\
       \end{pmatrix}.
$$
The rotation matrix then looks different as well as other objects we will meet later but the physical consequences are equivalent. By working with Eq.~(\ref{eq:Pmatrix}) we employed the convention of~\cite{Wein}.}.

We write down the usual four Bell states
\begin{subequations}\label{eq:Bellstates}
\begin{align}\label{eq:BellPhipm}
    \ket{\Phi(p_1,p_2)}_\pm&={1\over\sqrt{2}}\left(\ket{{p_1},+}\ket{{p_2},+}\pm\ket{{p_1},-}\ket{{p_2},-}\right)\\
    \label{eq:BellPsipm}
    \ket{\Psi(p_1,p_2)}_\pm&={1\over\sqrt{2}}\left(\ket{{p_1},+}\ket{{p_2},-}\pm\ket{{p_1},-}\ket{{p_2},+}\right),
\end{align}
\end{subequations}
where $\ket{{p_i},\pm}=a^\dg(p_{i},\pm)\ket{vac}$ using the notation of Appendix~B (from now on, the Latin subscripts label corresponding four-vectors and not their components). No need to repeat that the states above are not wave packets yet.
Using Eqs.~(\ref{eq:Pmatrix}) and~(\ref{eq:helrotation}) we rewrite states~(\ref{eq:BellPhipm}) and (\ref{eq:BellPsipm}). The helicity triads then read
\begin{align}\label{eq:mom-heleigsrepresented1}
    \ket{+}_ p&\equiv\ket{{p},+}={1\over{2}}\begin{pmatrix}
                  (\cos\vt+1)\exp{(-i\vp)} \\
                   (\cos\vt-1)\exp{(i\vp)} \\
                   -\sqrt{2}\sin\vt \\
                \end{pmatrix}\\
    \label{eq:mom-heleigsrepresented2}
      \ket{-}_ p&\equiv\ket{{p},-}={1\over{2}}\begin{pmatrix}
                   (\cos\vt-1)\exp{(-i\vp)}  \\
                  (\cos\vt+1)\exp{(i\vp)}  \\
                  -\sqrt{2}\sin\vt \\
                \end{pmatrix}.
\end{align}

An arbitrary rotation can be decomposed into  rotations around the $z$ axis followed by another rotation around the $y$ axis. Following the prescription sketched in Remark~\ref{rem:wignerangle} (or looking into~\cite{ging} or eventually~\cite{caban} where the Wigner angle was calculated in the representation of the cover of the Lorentz group -- the group $SL(2,\C)$) we find that the rotation $R_z(\la)$ around the $z$ axis through angle $\la$ induces the Wigner phase angle $\t_{{z}, p}\equiv\t_ p=0$ for all $p$ as long as $p\not=k$. Exactly the measure zero possibility $p=k$ is excluded from the definition of $\ket{\Phi(p_1,p_2)}$ because that would imply $p_1=p_2\equiv k$ and $\ket{\Phi(p_1,p_2)}$ to be converted  into $\ket{\phi}\propto\big([a^\dg(k,+)]^2+[a^\dg(k,-)]^2\big)\ket{vac}$ due to indistinguishability of such photons. Hence the Bell states transform as
\begin{align}\label{eq:ZrotatedPhipm}
    U(R_z(\la)):\ket{\Phi(p_1,p_2)}_\pm&\to\ket{\Phi(R_zp_1,R_zp_2)}_\pm\\
    \label{eq:ZrotatedPsipm}
    U(R_z(\la)):\ket{\Psi(p_1,p_2)}_\pm&\to\ket{\Psi(R_zp_1,R_zp_2)}_\pm.
\end{align}
The rotation $R_y(\varpi)$ around the $y$  axis results in a more complicated formula for the Wigner angle
\begin{equation}\label{eq:wignerangle}
    \t_{ p}=\arctan{\left[
    {\sin\varpi\sin\vp\over\sin\varpi\cos\vt\cos\vp+\cos\varpi\sin\vt}
    \right]}.
\end{equation}
To tackle the momentum dependence we will assume the momenta $ p_1, p_2$ to be correlated  such that the net effect of the rotation will be the overall phase equal to zero for all $ p_1$. That implies for Eq.~(\ref{eq:BellPhipm}) to find such $p_2$ that the corresponding Wigner phases are conjugated (because the helicity is equal for each product). Having $ p_i=(1,\sin\vt_i\cos\vp_i,\sin\vt_i\sin\vp_i,\cos\vt_i)$ this condition implies that we have to find solutions of the equation
\begin{equation}\label{eq:fixedpointsofPhipm}
    {\sin\varpi\sin\vp_1\over\sin\varpi\cos\vt_1\cos\vp_1+\cos\varpi\sin\vt_1}
    =-{\sin\varpi\sin(\vp_1+x)\over\sin\varpi\cos(\vt_1+y)\cos(\vp_1+x)+\cos\varpi\sin(\vt_1+y)},
\end{equation}
where $\vp_1+x=\vp_2,\vt_1+y=\vt_2$. This equation deserves an explanation. To find a solution means to find $x,y$ independent on $\varpi$. In other words, we want to locate all fixed points. Note that these points are not completely fixed since they vary with $\vp_1,\vt_1$.

Similarly, for the second couple Eq.~(\ref{eq:BellPsipm}) we have to find $ p_1, p_2$ such that the Wigner phases are equal (the helicity is opposite for each product)
\begin{equation}\label{eq:fixedpointsofPsipm}
    {\sin\varpi\sin\vp_1\over\sin\varpi\cos\vt_1\cos\vp_1+\cos\varpi\sin\vt_1}
    ={\sin\varpi\sin(\vp_1+x)\over\sin\varpi\cos(\vt_1+y)\cos(\vp_1+x)+\cos\varpi\sin(\vt_1+y)}.
\end{equation}

\begin{figure}[h]
  \begin{center}
  \hfill
  \begin{minipage}[t]{.5\textwidth}
    \begin{center}
      \epsfig{file=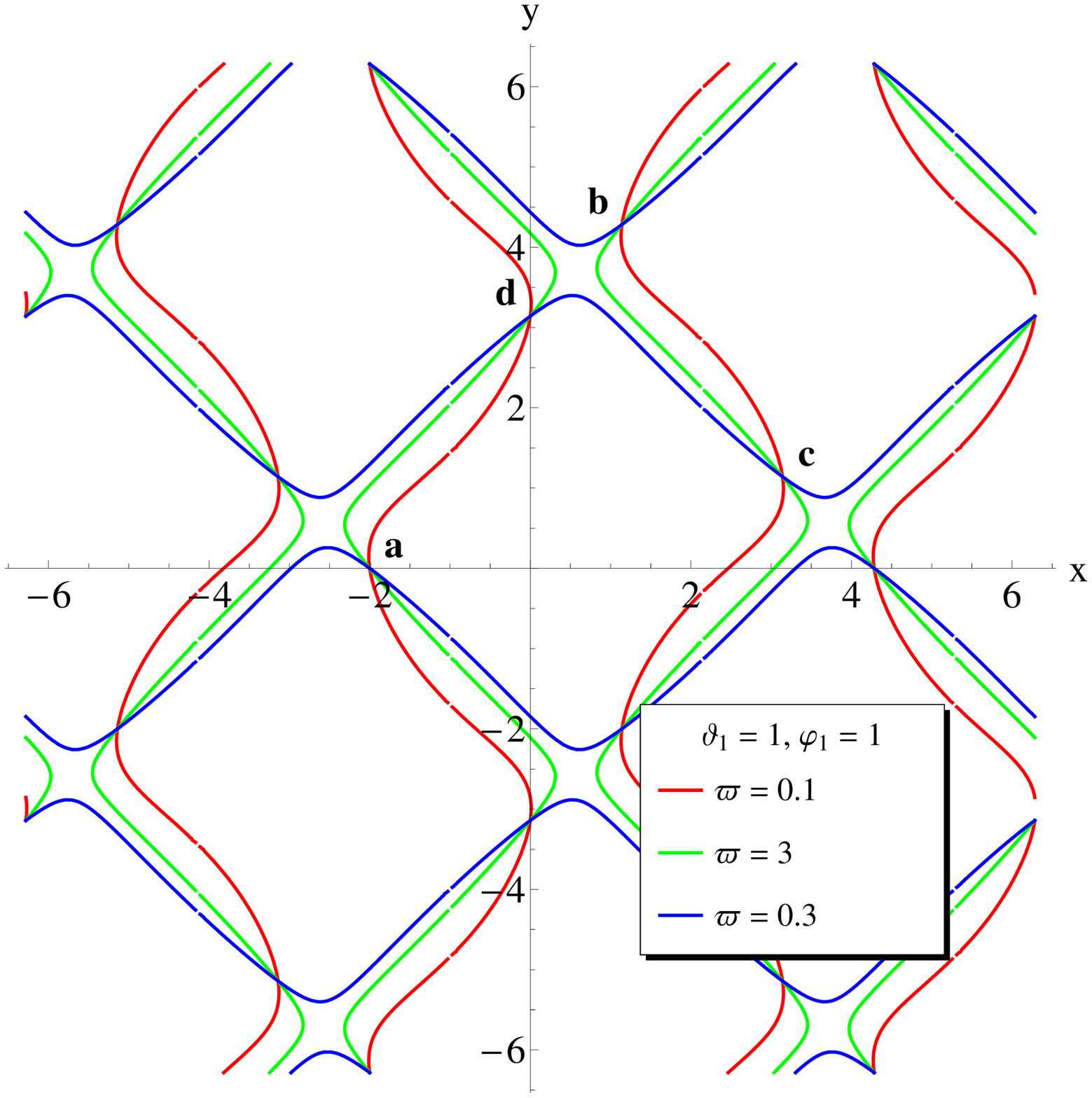, scale=0.4}
      \caption{Four fixed points corresponding to Eqs.~(\ref{eq:Phipma}),~(\ref{eq:Phipmb}), (\ref{eq:Phipmc}) and~(\ref{eq:Phipmd}) are indicated by points $a,b,c$ and $d$. Curves corresponding to a different rotation angle $\varpi$ about the $y$ axis will always meet in these points. The fixed points satisfy Eq.~(\ref{eq:fixedpointsofPhipm}). One can notice several gaps on every curve. These are singularities of Eq.~(\ref{eq:fixedpointsofPhipm}).}
      \label{fig:fixedpointsPhi}
    \end{center}
  \end{minipage}
  \hfill
  \begin{minipage}[t]{.47\textwidth}
    \begin{center}
      \epsfig{file=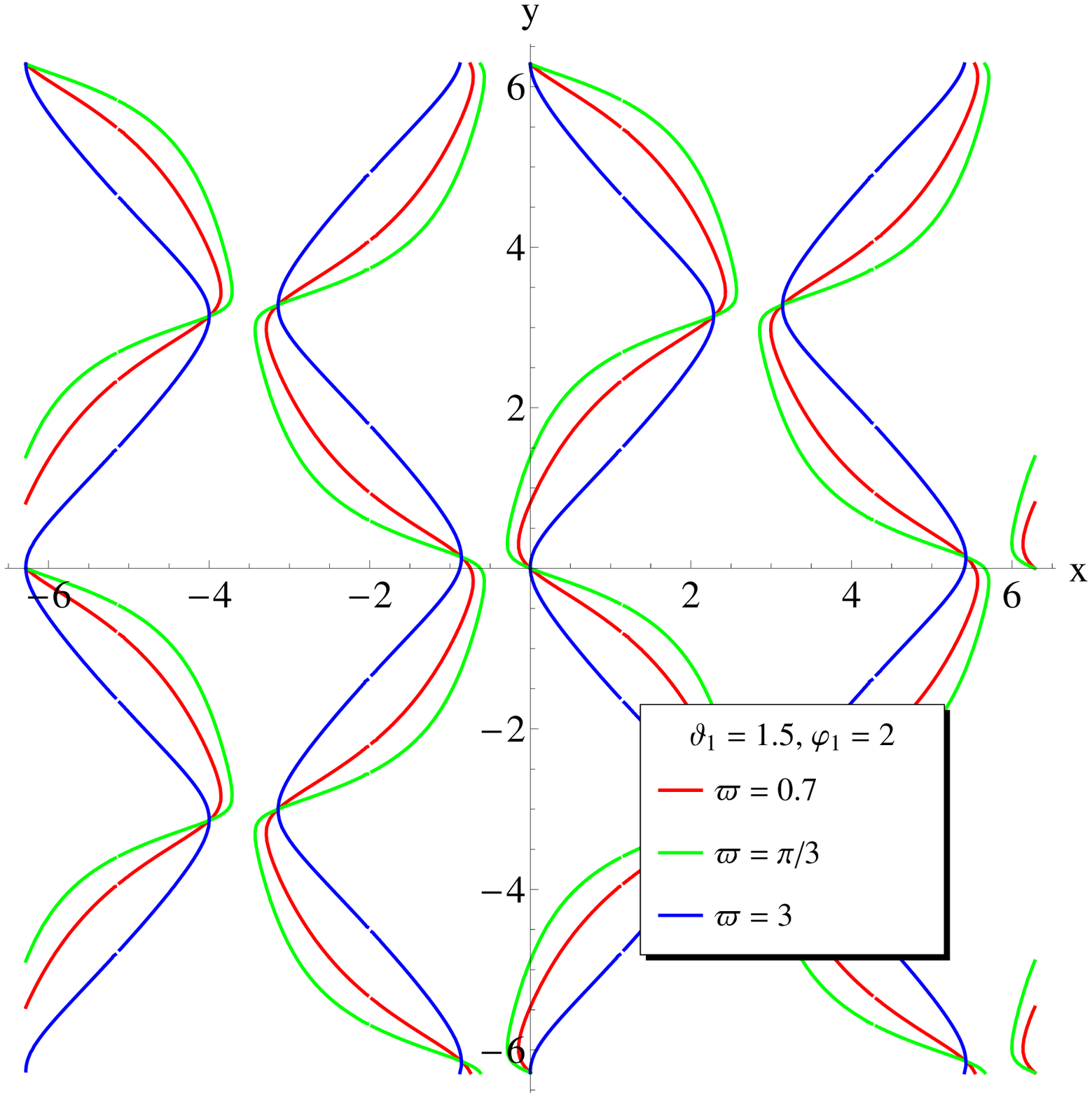, scale=0.4}
      \caption{Here we demonstrate how the position of fixed points changes for a different momentum direction. The fixed points satisfy Eq.~(\ref{eq:fixedpointsofPsipm}) so in contrast to Fig.~\ref{fig:fixedpointsPhi} there is a global (trivial) fixed point in the centre which is constant for all $\vt_1,\vp_1$. One can notice several gaps on every curve. These are singularities of Eq.~(\ref{eq:fixedpointsofPsipm}).}
      \label{fig:fixedpointsPsi}
    \end{center}
  \end{minipage}
  \hfill
    \end{center}
    \label{fig:fixedpoints1}
\end{figure}
\begin{lem}
    There exist four fixed points for Eqs.~(\ref{eq:fixedpointsofPhipm}) and~(\ref{eq:fixedpointsofPsipm}).
\end{lem}
\begin{proof}
    Eq.~(\ref{eq:fixedpointsofPhipm}) can be cast into the form
    \begin{equation}
        \tan{\varpi}\left[\sin{\vp_1}\cos{(\vt_1+y)}\cos{(\vp_1+x)}+\sin{(\vp_1+x)\cos{\vt_1}\cos{\vp_1}}\right]=-\sin{(\vp_1+x)}\sin{\vt_1}-\sin{\vp_1}\sin{(\vt_1+y)}.
    \end{equation}
    To have solutions independent on $\varpi$ and satisfying the above equation the only possibility is when both sides are equal to zero. That implies
    \begin{align}
        {\cos{(\vt_1+y)}\cos{(\vp_1+x)}\over\sin{(\vp_1+x)}}&=-{\cos{\vt_1}\cos{\vp_1}\over\sin{\vp_1}}\\
        {\sin{(\vt_1+y)}\over\sin{(\vp_1+x)}}&=-{\sin{\vt_1}\over\sin{\vp_1}}.
    \end{align}
    It is easy to enumerate all possibilities for which, for example, the first equation is satisfied and then to check if it also holds  for the second one. We end up with four different solutions (modulo ${2\pi}$)
    \begin{subequations}
\begin{align}
     &\vp_2=-\vp_1,\vt_2=\vt_1\label{eq:Phipma}\\
     &\vp_2=\pi-\vp_1,\vt_2=-\vt_1\label{eq:Phipmb}\\
     &\vp_2=\pi+\vp_1,\vt_2=\pi-\vt_1\label{eq:Phipmc}\\
     &\vp_2=\vp_1,\vt_2=\vt_1+\pi\label{eq:Phipmd}.
\end{align}
\end{subequations}
Having the solutions of Eq.~(\ref{eq:fixedpointsofPhipm}) we immediately  get the solutions for Eq.~(\ref{eq:fixedpointsofPsipm}). The later equation is obtained from the former one by a mere shift $x\mapsto x,y\mapsto y+\pi$. Hence
\begin{subequations}
\begin{align}
     &\vp_2=-\vp_1,\vt_2=\vt_1+\pi  \label{eq:Psipma}\\
     &\vp_2=\pi-\vp_1,\vt_2=\pi-\vt_1   \label{eq:Psipmb}\\
     &\vp_2=\pi+\vp_1,\vt_2=-\vt_1\label{eq:Psipmc}\\
     &\vp_2=\vp_1,\vt_2=\vt_1\label{eq:Psipmd}.
\end{align}
\end{subequations}
In Figs.~\ref{fig:fixedpointsPhi} and~\ref{fig:fixedpointsPsi} we see the position of four fixed points for some $\vt_1,\vp_1$. If we had in Fig.~\ref{fig:fixedpointsPsi} the same parameters $\vt_1,\vp_1$ as in Fig.~\ref{fig:fixedpointsPhi} the picture would indeed be only shifted along the $y$ axis. We can check that all the solutions satisfy the corresponding Eqs.~(\ref{eq:fixedpointsofPhipm}) and~(\ref{eq:fixedpointsofPsipm}).
\end{proof}

The first solution Eq.~(\ref{eq:Phipma}) is the most interesting one since $ p_2$ dwells in the same hemisphere as $ p_1$. Commenting on the rest of the solutions, Eq.~(\ref{eq:Phipmb}) is equivalent to Eq.~(\ref{eq:Phipma}) and the remaining two solutions are a bit awkward because $ p_2$ points in the opposite direction. We might, however, say that
\begin{equation}\label{eq:YrotatedPhipm}
    U(R_y(\varpi)):\ket{\Phi(p_1,p_2)}_\pm\to\ket{\Phi(R_yp_1,R_yp_2)}_\pm.
\end{equation}

From the physical point of view there is a noteworthy asymmetry in the solutions  for Eq.~(\ref{eq:fixedpointsofPsipm}). The first two fixed points again correspond to a wave packet with the second momentum pointing in the opposite direction and the last two solutions are just an identity operation (${ p_2}={ p_1}$). So there are two solutions where
\begin{equation}\label{eq:YrotatedPsipm}
    U(R_y(\varpi)):\ket{\Psi(p_1,p_2)}_\pm\to\ket{\Psi(R_yp_1,R_yp_2)}_\pm
\end{equation}
holds but, for the practical purposes, the opposite momentum direction is an obvious obstacle.

The conclusion is that as long as we keep the momenta in the Bell states correlated as described above it is not necessary to introduce three-dimensional helicity density matrices since there is no phase change due to the spatial rotation at all (cf. Eqs.~(\ref{eq:ZrotatedPhipm}),~(\ref{eq:ZrotatedPsipm}),~(\ref{eq:YrotatedPhipm}) and ~(\ref{eq:YrotatedPsipm})).

The way of getting rid of the Wigner phase for spatial rotations is  reminiscent of the method presented in~\cite{relinv}. The authors constructed wave packets by entangling two spatially distinguishable wave packets carrying the same momentum and opposite helicities. Every spatial rotation induces an opposite Wigner phase for each wave packet so the resulting phase is zero~\footnote{Note that this construction is not working for arbitrary wave packets. It works only for wave packets where the rotation induces the same Wigner phase for all momenta from which the wave packet is `assembled'.}. As we already stressed, however, the Bell states we work with so far are not wave packets so the construction presented here is qualitatively different.

\subsection{The effect of Peres-Terno measurement}\label{sec:measurement}

The projection present in the Peres-Terno measurement process (Eq.~(\ref{eq:projpol})) acts by `cutting off' the $z$ component of Eqs.~(\ref{eq:mom-heleigsrepresented1}) and~(\ref{eq:mom-heleigsrepresented2}). Henceforth, after the proper normalization we get
\begin{align}\label{eq:mom-heleigscutoff1}
    \Pi_ k: \ket{{p},+}\to\ketf{p,+}&={1\over N}\begin{pmatrix}
                  (\cos\vt+1)\exp{(-i\vp)} \\
                   (\cos\vt-1)\exp{(i\vp)} \\
                \end{pmatrix}\\
     \label{eq:mom-heleigscutoff2}
     \Pi_ k: \ket{{p},-}\to\ketf{p,-}&={1\over N}\begin{pmatrix}
                  (\cos\vt-1)\exp{(-i\vp)} \\
                   (\cos\vt+1)\exp{(i\vp)} \\
                \end{pmatrix},
\end{align}
where $N={\sqrt{2+2\cos^2{\vt}}}$. Note that $\brk{p,+}{p,-}=0$ for all $p$ but $\bkf{p,+}{p,-}\not=0$ as expected. The `floored' kets indicate the action of the projection $\Pi_ k$ on a momentum-helicity state.

Interesting things start to happen when we ask how the projection $\Pi_ k$ acts on the Bell states. First of all, these states are two-mode states so the projection is actually $\Pi^{(1)}_ k\otimes\Pi^{(2)}_ k$ where the superscripts indicate which mode is being measured. The action of the projection can be found in Appendix~A. In general, these states are subnormalized and non-orthogonal and thus resembling a general situation studied in~\cite{pete03}. But, intriguingly, when plugging in some of the solutions for the momentum correlations (namely Eqs.~(\ref{eq:Phipma}),~(\ref{eq:Phipmc}),~(\ref{eq:Phipmd}),~(\ref{eq:Psipma}) and~(\ref{eq:Psipmb})) found from the study of fixed points for arbitrary rotations about the $y$ axis  we find that  (i) the resulting floored Bell states are orthogonal for all $\vt_1,\vp_1$ and (ii) one of the states is always (up to a normalization) invariant. This has some interesting consequences particularly for the solution Eq.~(\ref{eq:Phipma}). Recall that this is the only non-trivial solution where the resulting state `points' in the same direction as the propagation direction. From now on we will stick just to this solution since it will be later relevant for our discussion of localized wave packets  (the analysis for the rest of the non-trivial solutions gives similar results but their use as localized wave packets is obviously none).

When the previous claim written in detail for $\vp_2=-\vp_1,\vt_2=\vt_1$ we get from Eqs.~(\ref{eq:APPPhiprojplus}) and~(\ref{eq:APPPhiprojminus})
\begin{multline}\label{eq:Phiplusproj}
     \Pi^{(1)}_ k\otimes\Pi^{(2)}_ k:\ket{\Phi(p_1,p_2)}_+\to\ketf{\Phi(p_1,p_2)}_+
     =\begin{pmatrix}
      1 \\
      \exp{(-2i\vp_1)}{\cos^2\vt_1-1\over\cos^2\vt_1+1} \\
      \exp{(2i\vp_1)}{\cos^2\vt_1-1\over\cos^2\vt_1+1} \\
       \,1\, \\
     \end{pmatrix}\\\xrightarrow{\rm normalization}a_1\ketf{\Phi_+}+a_2\ketf{\Psi_+}+ia_3\ketf{\Psi_-}
\end{multline}
\begin{equation}\label{eq:Phiminusproj}
     \Pi^{(1)}_ k\otimes\Pi^{(2)}_ k:\ket{\Phi(p_1,p_2)}_-\to\ketf{\Phi(p_1,p_2)}_-
     ={1\over1+\cos^2{\vt_1}}
     \begin{pmatrix}
     2\cos\vt_1 \\
     0 \\
     0 \\
      -2\cos\vt_1 \\
     \end{pmatrix}\xrightarrow{\rm normalization}\ketf{\Phi_-},
\end{equation}
where $a_i\in\R$ and $\sum_ia_i^2=1$. We denoted $\ketf{\Phi_+}=1/\sqrt{2}(\ketf{0_1}\ketf{0_2}+\ketf{1_1}\ketf{1_2})$ and similarly for the rest of the states. Note that both output states are orthogonal for all $\vt_1$ and $\vp_1$. Also, the resulting states are subnormalized even though states~(\ref{eq:mom-heleigscutoff1}) and~(\ref{eq:mom-heleigscutoff2}) are normalized. This is again the effect of cutting off (alias projection) the `longitudinal' part of Eqs.~(\ref{eq:mom-heleigsrepresented1}) and~(\ref{eq:mom-heleigsrepresented2}). The normalization factor for Eq.~(\ref{eq:Phiplusproj}) and Eq.~(\ref{eq:Phiminusproj}) is ${1+\cos^2\vt_1\over2\sqrt{1+\cos^4\vt_1}}$ and ${1+\cos^2\vt_1\over\sqrt{8}{\cos\vt_1}}$, respectively. Recall that the measurement indeed acts as a trace-decreasing non-completely positive map~\cite{nonCP}. We also want to draw attention to the similarity between the behavior of the investigated Bell states under the irrep of the product $SU(2)\times SU(2)$ and the projection $\Pi^{(1)}_ k\otimes\Pi^{(2)}_ k$. In the fist case the singlet is preserved and the triplet states transform among themselves. This is a starting point to considerations about reference frames and decoherence-free subspaces~\cite{refframesReview}. In our case one of the triplets is invariant meanwhile the second triplet mingles with the singlet and the remaining triplet state.

An important comment also related to decoherence-free subspaces has to made. As we have seen, the Bell states remain orthogonal after the projection
$$
\brk{\Phi(p_1,p_2)_+}{\Phi(p_1,p_2)_-}=0
\xrightarrow{\ \Pi^{(1)}_ k\otimes\Pi^{(2)}_ k}
\bkf{\Phi(p_1,p_2)_+}{\Phi(p_1,p_2)_-}=0.
$$
However, the action of the projection is not $SU(2)\times SU(2)$ covariant. In other words, there is a preferred basis spanned by states from Eq.~(\ref{eq:Bellstates}). One faces a slightly analogous situation encountered in the study of decoherence-free subspaces~\cite{refframe2003}. However, the presented results still hold for any orthogonal basis $\ket{{p_i},+},\ket{{p_i},-}$ forming the actual entangled states in Eq.~(\ref{eq:Bellstates}). The only difference is given by the fact that our underlying Hilbert space is infinite-dimensional.

\subsection{Boost along the $z$ axis followed by measurement}
The found result has several serious consequences. First, if we prepare two orthogonal pure qubits  having the basis formed by  the Bell pair $\ket{\Phi(p_1,p_2)}_\pm$, they will be perfectly distinguishable by the measurement. This was not the case for the original Peres-Terno scheme. But more important, if we boost along the $z$ axis and then perform the same projective measurement we see that state~(\ref{eq:Phiplusproj}) remains invariant in the sense that the boost only modifies $\vt_1(=\vt_2)$ and keeps $\vp_1(=-\vp_2)$ intact. Of course, by the change of $\vt_1,\vt_2$ the four-momenta $ p_1, p_2$ change too as Eq.~(\ref{eq:particle}) dictates. They symmetrically close in or open up depending on the observer's velocity direction. In a real experiment this would be registered as the color change (this is the Doppler effect which we formally suppress by the renormalization of the new four-vectors).

Hence the effect of a boost in the $z$ direction $B_z(\eta)$ followed by the projective measurement $\Pi^{(1)}_ k\otimes\Pi^{(2)}_ k$ is
\begin{align}\label{eq:Phiplusboosted}
     \Pi^{(1)}_ k\otimes\Pi^{(2)}_ k:\ket{\Phi(B_zp_1,B_zp_2)}_+&\to\ketf{\Phi(B_zp_1,B_zp_2)}_+\\
    \label{eq:Phiminusboosted}
     \Pi^{(1)}_ k\otimes\Pi^{(2)}_ k:\ket{\Phi(B_zp_1,B_zp_2)}_-&\to\ketf{\Phi(B_zp_1,B_zp_2)}_-.
\end{align}

Let us conclude this section with two comments. (i) As in the previous case all output states are orthogonal for all boosts and every $ p_1$. We have to stress, however, that the above transformations are not unitary. We recall that the normalization of Eqs.~(\ref{eq:mom-heleigscutoff1})-(\ref{eq:Phiminusproj}) is explicitly $\vt-$dependent. It would be a clear sign of inconsistency if we had a finite-dimensional unitary representation for boost transformations -- the boost generators are the noncompact part of the Lorentz algebra. Similar inconveniences were encountered in~\cite{pete03} but in our case it does not cause serious troubles. (ii) For a state propagating in the $z$ direction the most general Lorentz transformation can be expressed as $R_z(\la)R_y(\varpi)B_z(\eta)$. Hence, we have shown that the Bell states propagating in the $z$ direction and with the appropriately correlated momenta are unaffected by an arbitrary rotation (and therefore a general boost) if measured according to the Peres-Terno scheme.

\subsection{Relativistically invariant photonic wave packets}\label{sec:relinvwavpack}

The real touchstone for quality will be the behavior of wave packets composed from the previously studied entangled states under the Lorentz transformation.  As indicated, we will demonstrate the construction and properties of invariant wave packets with $\ket{\Phi(p_1,p_2)}_\pm$  serving as the logical basis for construction of a (pure) wave packet qubit.  We will not have anything specific to say about the actual feasibility of producing of such states in a laboratory. Also, the dispersion effects of other than relativistic origins are not considered here.

Let us follow the prescription and notation for the single-mode case~Eq.~(\ref{eq:genwavepacket})
\begin{equation}\label{eq:genwavpackBell}
    \ket{\Omega_{12}}=\int{\rm d}\mu( p_1){\rm d}\mu( p_2)f( p_1, p_2)\big(\a_{ p_{1,2}}\ket{\Phi(p_1,p_2)}_++\b_{ p_{1,2}}\ket{\Phi(p_1,p_2)}_-\big).
\end{equation}
Similarly to the situation of ordinary wave packets we assume $\a$ and $\b$ to be momentum independent (see the discussion below  Eq.~(\ref{eq:genwavepacket_rewritten})).

Taking into account Eqs.~(\ref{eq:ZrotatedPhipm}) and~(\ref{eq:YrotatedPhipm}) and how the boost acts we easily find
\begin{align}\label{eq:genwavpackBell_trabsformed}
    \ket{\mho_{12}}&=U(R_z(\la)R_y(\varpi)B_z(\eta))\ket{\Omega_{12}}\nn\\
    &=\a\int{\rm d}\mu( q_1){\rm d}\mu(q_2)f(q_1,q_2)\ket{\Phi(q_1,q_2)}_+
    +\b\int{\rm d}\mu( q_1){\rm d}\mu(q_2)f(q_1,q_2)\ket{\Phi(q_1,q_2)}_-,
\end{align}
where $q_i=R_z(\la)R_y(\varpi)B_z(\eta) p_i$. As discussed earlier in Eq.~(\ref{eq:wavpackrotated}) the measure and the envelope function are Lorentz invariant ($f(q_1,q_2)\equiv f(p_1,p_2)$). Hence, we call the transformed wave packet invariant because the state amplitudes remain unaltered by an arbitrary Lorentz transformation. In other words, otherwise disastrous decoherence between helicity and momentum degrees of freedom is kept under control by a careful choice of the logical basis.

As a final step,  the measurement described by $\Pi^{(1)}_ k\otimes\Pi^{(2)}_ k$  is realized, due to Eqs.~(\ref{eq:Phiplusproj}) and~(\ref{eq:Phiminusproj}), by the projections
\begin{equation}\label{eq:wavpackprojs}
    \Gamma_1=\kbf{\Phi_-}{\Phi_-},\quad\Gamma_2=\openone-\Gamma_1.
\end{equation}
To be more precise, we have seen that the map related to the measurement operator $\Pi^{(1)}_ k\otimes\Pi^{(2)}_ k$ is trace-decreasing. Consequently, this measurement is a conditional (or post-selected) measurement. The overall probability of measurement is lower than one  and depends on the details of the wave packet, namely on the choice of its envelope function. When the measurement event occurs the observer can always perfectly distinguish both orthogonal basis states since the two subspaces from Eq.~(\ref{eq:wavpackprojs}) are orthogonal and independent on the observer's reference frame.  Note that we tacitly but naturally assume the existence of a global coordinate system.


\section{Conclusions}

There exists a class of photonic states we call the realistic wave packets. It is a weighted and normalized superposition of common eigenvectors of the momentum and helicity operators. The reason we call them realistic is twofold. From the mathematical point of view, if the corresponding envelope function is square integrable (but otherwise almost arbitrary) function, such an object lives in a Hilbert space. From the physical point of view, this object can  in principle be prepared, transmitted, manipulated in a localized manner and finally detected in a laboratory. The problem appears when we consider wave packets in a relativistic regime. The general Lorentz transformation related to the change of a reference frame acts individually on every momentum and induces momentum-helicity entanglement. The situation is even worse when we try to recover the information by an act of measurement. The simplest (and probably the only feasible) way is the Peres-Terno scheme of the helicity projection onto the plane perpendicular to the direction of propagation.

Surprisingly, in this paper we have shown that under these unfavorable conditions we are still able to prepare  localized wave packets for which the information can be, at least in principle, transmitted and recovered in every reference frame with perfect fidelity. There are two key ingredients. First, the logical qubit basis are two-mode maximally entangled helicity states. Second, the momentum degrees of freedom of these states are correlated in a precise manner. This results in canceling the explicit momentum dependence for an arbitrary rotation and boost. Intriguingly, the same condition implies that during the measurement according to the above scheme the orthogonal logical basis states forming a qubit are projected into two orthogonal subspaces enabling us to perfectly distinguish them.

It has been known from the previous analysis that the map governing the evolution and measurement of generic wave packets in the relativistic setting is trace-decreasing and non-completely positive. It is a consequence of the non-covariant helicity measurement scheme and this leads to all aforementioned effects and troubles. Here we have encountered precisely  the same behavior but contrary to the previous work these effects are now harmless. Consequently, using the investigated encoding the problem of non-covariant transformation of the von Neumann entropy of helicity density matrices disappears.

Finally, let us stress that the choice of the Peres-Terno measurement (arguably spoiling all nice transformation properties of the momentum-helicity eigenstates) is directly motivated by its relative feasibility in a laboratory. Should more sophisticated helicity measurements be available we might obtain more elegant results.

\section*{Acknowledgements}

It is a pleasure to thank Roc\'io J\'auregui and Prakash Panangaden for comments on the draft. The work was supported  by a  grant from the Office of Naval Research  (N000140811249).

\section*{Appendix A}
\renewcommand{\theequation}{A-\arabic{equation}}
\setcounter{equation}{0}

 The projection corresponding to the Peres-Terno measurement scheme acts as
\begin{equation}\label{eq:projectoraction}
    \Pi^{(1)}_k\otimes\Pi^{(2)}_k:\ket{{p_i},\s_i}\ket{{p_j},\s_j}\to\ketf{{p_i},\s_i}\ketf{{p_j},\s_j},
\end{equation}
where the `floored' quantities indicate the projection action. Denoting
$$
{1\over N}={1\over\sqrt{(1+\cos^2{\vt_1})(1+\cos^2{\vt_2})}}\\
$$
we get
\begin{multline}\label{eq:APPPhiprojplus}
    \Pi^{(1)}_k\otimes\Pi^{(2)}_k:\ket{\Phi(p_1,p_2)}_+\to\ketf{\Phi(p_1,p_2)}_+
    ={1\over N}\begin{pmatrix}
    (\cos\vt_1\cos\vt_2+1)\exp{(-i(\vp_1+\vp_2))}\\
    (\cos\vt_1\cos\vt_2-1)\exp{(-i(\vp_1-\vp_2))}\\
    (\cos\vt_1\cos\vt_2-1)\exp{(i(\vp_1-\vp_2))}\\
    (\cos\vt_1\cos\vt_2+1)\exp{(i(\vp_1+\vp_2))}\\
    \end{pmatrix},
\end{multline}
\begin{multline}\label{eq:APPPhiprojminus}
    \Pi^{(1)}_k\otimes\Pi^{(2)}_k:\ket{\Phi(p_1,p_2)}_-\to\ketf{\Phi(p_1,p_2)}_-
    ={1\over N}\begin{pmatrix}
    (\cos\vt_1+\cos\vt_2)\exp{(-i(\vp_1+\vp_2))}\\
    -(\cos\vt_1-\cos\vt_2)\exp{(i(\vp_2-\vp_1))}\\
    (\cos\vt_1-\cos\vt_2)\exp{(-i(\vp_2-\vp_1))}\\
   -(\cos\vt_1+\cos\vt_2)\exp{(i(\vp_1+\vp_2))}\\
    \end{pmatrix},
\end{multline}
\begin{multline}\label{eq:APPPsiprojplus}
    \Pi^{(1)}_k\otimes\Pi^{(2)}_k:\ket{\Psi(p_1,p_2)}_+\to\ketf{\Psi(p_1,p_2)}_+
    ={1\over N}\begin{pmatrix}
    (\cos\vt_1\cos\vt_2-1)\exp{(-i(\vp_1+\vp_2))}\\
    (\cos\vt_1\cos\vt_2+1)\exp{(i(\vp_2-\vp_1))}\\
    (\cos\vt_1\cos\vt_2+1)\exp{(-i(\vp_2-\vp_1))}\\
    (\cos\vt_1\cos\vt_2-1)\exp{(i(\vp_1+\vp_2))}\\
    \end{pmatrix}
\end{multline}
and
\begin{multline}\label{eq:APPPsiprojminus}
    \Pi^{(1)}_k\otimes\Pi^{(2)}_k:\ket{\Psi(p_1,p_2)}_-\to\ketf{\Psi(p_1,p_2)}_-
    ={1\over N}\begin{pmatrix}
   - (\cos\vt_1-\cos\vt_2)\exp{(-i(\vp_1+\vp_2))}\\
    (\cos\vt_1+\cos\vt_2)\exp{(-i(\vp_1-\vp_2))}\\
   -(\cos\vt_1+\cos\vt_2)\exp{(i(\vp_1-\vp_2))}\\
    (\cos\vt_1-\cos\vt_2)\exp{(i(\vp_1+\vp_2))}\\
    \end{pmatrix}.
\end{multline}

\section*{Appendix B - Review the unitary representation of the Poincar\'e group}
\renewcommand{\theequation}{B-\arabic{equation}}
\setcounter{equation}{0}

The purpose of the following two Appendices is to recall the problem of the polarization encoding for realistic photonic wave packets. Readers not familiar with this problem full of intricate details might find this short review useful.

The Poincar\'e algebra has two Casimir operators $C_1=P_\chi P^\chi$ and $C_2=W_\chi W^\chi$ where $W^\chi=-1/2\,\e^{\chi\rho\a\omega}P_\rho J_{\a\omega}$ is the Pauli-Lubanski vector, $\{J_{23},J_{31},J_{12}\}\equiv\{J_1,J_2,J_3\}$ are the total angular momentum generators and $\{J_{01},J_{02},J_{03}\}\equiv\{K_1,K_2,K_3\}$ are the boost generators. We recall that $\e^{\chi\rho\a\omega}$ is the Levi-Civita tensor and $\e^{0123}=-1$. The corresponding Poincar\'e group induces translations $x^\mu\to \La^\mu_\nu x^\nu+a^\mu$ (in the coordinate representation) where $\La(\zeta)\in SO(3,1)$ is a Lorentz transformation $\La(\zeta)=\exp{(-i/2\zeta^{\s\vs}J_{\s\vs})}$ and $T(a)=\exp{(-ia^\mu P_\mu)}$. It is well known that the Poincar\'e group is a semidirect product of the proper orthochronous Lorentz group and a group of translations $SO(3,1)\ltimes R^4$.

The group acts on set $K$  which is a linear vector space equipped with the Minkowski metric with the signature $\{-+++\}$. Hence the set $K$ is a manifold rich in structure. Note that in the case of Minkowski spacetime  the action of the Poincar\'e group is uninteresting since it acts transitively. It is not so in the four-momentum space. The Lorentz group acts in a similar manner for both spaces but the action of the translation operator differs. One can see it easily (in an admittedly handwaving way) realizing that the Fourier transform of a function $f:x_i\mapsto x_i-a_i$ for some constant $a$ is a simple phase transformation in the three-momentum space. 

The analysis of the four-momentum space $K$ shows that there are six orbits. The whole classification from the physical point of view can be found, for instance, in~\cite{Wein,Tung}.  We are interested in the orbit which physically corresponds to null vectors with a positive zero component of the momentum four-vector. They correspond (after quantization) to the scarce but important class of free massless particles containing photons and possibly gravitons. By the very definition of orbit no Lorentz transformation can change the character of a null vector. In other words, photons always travel at the speed of light no matter the reference frame a potential observer resides in. We are free to generate the `null orbit'  from any null momentum four-vector but it is customary and advantageous to choose the simplest possible one -- the standard four-momentum $k^\mu=(1,0,0,1)$. Note that other choices are completely equivalent as discussed in detail in~\cite{halpern}. From the physical perspective, the standard direction corresponds to (eventually quantized) plane electromagnetic waves (photons) traveling at the normalized speed of light along the $z$ axis. Nevertheless, the word `physical' should be used carefully. It is not, in principle, possible to generate such a field (classical or quantum). The reason is that the sharp momentum value implies that the classical wave is `infinitely' spread in position and one would need an infinite amount of energy to prepare such a field. In the quantum case the corresponding state does not even occupy a Hilbert space of square integrable functions and is completely delocalized in Minkowski spacetime. Classical or quantum fields created in a laboratory are in reality wave packets with a finite spread in momentum and position.

Using another property of the algebra $[P^\mu,P^\nu]=0$ and $C_1=C_2=0$ which holds in the massless case only we denote a `single-particle' state $\ketp$ as an eigenstate of $P^\mu$ where $p^\mu$ are eigenvalues of $P^\mu$. Other possible degrees of freedom related to $W^\mu$  rather than to $C_2$ (note that $[W^\mu,P^\nu]=0$) are gathered in $\s$
\begin{equation}\label{eq:particle}
    P^\mu\ketp=p^\mu\ketp.
\end{equation}

To be able to use quantum mechanics in the relativistic context as we habitually do in ordinary quantum mechanics we need to introduce a Hilbert space and a unitary representation in it which respects the composition law for the Poincar\'e group
\begin{equation}\label{eq:unitaryrepre}
    U(\{\tilde\La,b\})U(\{\La,a\})=\exp{(i\Upsilon)}U(\{\La\bar\La,\tilde\La a + b\}).
\end{equation}
The appearance of the phase is due to the projective nature of quantum states (rays instead of vectors). For simply-connected groups it is possible  to get rid of the phase factor ($\Upsilon=0$). This is not the case of the Lorentz (or Poincar\'e) group and it will have some impact at a later stage.

The Poincar\'e group is non-compact and therefore there is no faithful finite-dimensional unitary representation. This might have been a huge obstacle for studying the evolution of relativistic quantum states but Wigner realized how to circumvent this problem. He found that the unitary action of the Lorentz group is governed by the irreps of its stabilizer subgroups~\cite{Tung}. Loosely speaking, the little group `induces' its action on the Lorentz group. The method is called the method of induced representations and was later studied and generalized by~\cite{mackey} and others. The Wigner method prescribes that the unitary action of a Lorentz transformation $\La_\mu^\nu p^\mu=p'^\nu$ reads
\begin{equation}\label{eq:wigneraction}
    U(\La)\ketp=\sum_{\s'}D_{\s,\s'}(S_ k)\ket{p',\s'},
\end{equation}
where $D_{\s,\s'}(S_ k)$ is the irrep of the stabilizer group $S_ k$ keeping the standard null direction $ k$ intact. The stabilizer group
can be written
\begin{equation}\label{eq:stabilizer}
    S_ k=L^{-1}(\La p)\La L( p),
\end{equation}
where $L( p): k\to p$ takes the standard four-vector to any other null vector (i.e., to a vector on the same orbit). It turns out that the stabilizer group for the standard null orbit is the Euclidean group $ISO(2)$ (also known as $E(2)$) transforming a two-dimensional (Euclidean) plane into itself. It is a three parametric group generated by translations in two directions and a rotation around the axis perpendicular to the plane. The components of the Pauli-Lubanski vector are the generators of the corresponding Lie algebra. The group is noncompact so it might seem that we did not improve our situation but we actually did. To make the long story short~\cite{Wein} it appears that the noncompact part of the stabilizer group is not physically relevant to the evolution of the photon under the Lorentz group so the only $SO(2)$ subgroup of $ISO(2)$ remains. The corresponding algebra generator is $W^0=W^3\equiv J_3$ and we find
\begin{equation}\label{eq:heli}
    J_3\ket{k,\s}=\s\ket{k,\s}.
\end{equation}
We can see from Eq.~(\ref{eq:heli}) that
\begin{equation}\label{eq:pheli}
    U(R(\vt,\vp))J_3U^{-1}(R(\vt,\vp))\ketp=\s\ketp
\end{equation}
so we define helicity as the projection of the angular momentum along the direction of motion and $\s$ is clearly a relativistic invariant.

The $SO(2)$ group element is as usual recovered by exponentiation $R_z(\t)=\exp{(-i\t J_3)}$ so for the purpose of Eq.~(\ref{eq:wigneraction}) we get
\begin{equation}\label{eq:Dmatrixforrot}
    D_{\s,\s'}(R_z(\t))=\exp{[-i\s\t]}\delta_{\s\s'}
\end{equation}
and so for an arbitrary Lorentz group $\La: p\to q$
\begin{equation}\label{eq:unitarylorentz}
    U(\La)\ketp=\exp{\left[-i\s\t_{ p, q}\right]}\ketq.
\end{equation}
\begin{rem}\label{rem:wignerangle}
    We will call $\t_{p,q}$ the Wigner angle and it can be explicitly calculated  for a given Lorentz transformation $\La: p\to q$ from Eq.~(\ref{eq:stabilizer}) by setting $R_z(\t_{p,q})=S_ k$. Notice that the phase factor $\exp{\left[-i\s\t_{p,q}\right]}\in U(1)$ is indeed a unitary representation of $R_z(\t_{ p, q})\in SO(2)$. The two groups are known to be isomorphic.
\end{rem}
\begin{rem}
The parameter $\s$ can take just integer or half-integer values as the result of the inability to set the phase in Eq.~(\ref{eq:unitaryrepre}) to zero~\cite{Wein}. Moreover,  the momentum projection in two opposite directions $ p$ and $- p$ are related by the parity operator so for photons as  massless vector bosons we get $\s=\pm1$. We see that it is necessary to enlarge the Poincar\'e group by including the parity transformation~\cite{arvmuk}. Clearly, there is no helicity equal to zero since there is no rest frame for a photon.
\end{rem}
\begin{rem}
    A state $\ketp$ coincides with the action of creation operators  $\crop\ket{vac}$ to the vacuum state from the usual quantization procedure of a free massless vector field yielding a single particle state of a momentum $ p$ and helicity $\s=\pm1$ (alias circular polarization). It is straightforward to generate multi-particle states by a repeated application of the creation operators for different modes satisfying $[\cropp,\crop]=\d( p'- p)\d_{\s'\s}$ and thus to create the familiar Fock space as a direct sum of the completely symmetric Hilbert space of $n$ photons $\F=\bigoplus_{n=0}^\infty\H_n^{\rm sym}$.

    Also notice that one of the consequences of the Wigner procedure  is a clear indication that there are just two `spin' degrees of freedom and they are always perpendicular to the direction of motion and to each other. We know this fact without mentioning the Coulomb gauge whatsoever. The Coulomb gauge can achieve the same goal but one has to pay the price of not having a manifestly covariant gauge condition. That is, the gauge condition must be imposed for every reference frame separately. On the other hand, the Coulomb gauge is important and useful in quantum field theory from a broader point of view.
\end{rem}

\section*{Appendix C - Measurement troubles and wave packet helicity density matrices}
\renewcommand{\theequation}{C-\arabic{equation}}
\setcounter{equation}{0}

The authors of~\cite{pete03} proposed a very simple measurement which corresponds to what actually might happen in a laboratory when measuring the helicity (polarization) degrees of freedom. Assume that the direction of propagation of the wave packet is the standard direction $k$. It is customary to place a polarization analyzer perpendicular to the direction of motion. Our helicity eigenstates $\ket{p,\pm}$ can be represented as two orthogonal complex four-vectors $\ve^\mu_{p,\pm}$. For a free field the Coulomb gauge $p_i\ve^i_{ p,\pm}=0$ implies the axial component of the helicity vector to be zero $\ve^0_{p,\pm}=0$ (holds for all $p$ because just rotations are considered now). Quantum electrodynamics indicates what happens next. We summon up the well-known relation for the helicity eigenvectors~\cite{Wein}
\begin{equation}\label{eq:projpol_general}
    \Pi_{p}^{ij}=\sum_{\s=\pm}\ve^i_{p,\s}\bar\ve^j_{p,\s}=\d^{ij}-{p^ip^j\over p_kp^k},
\end{equation}
where the overbar indicates complex conjugation. Following the notation of Eqs.~(\ref{eq:mom-heleigsrepresented1}) and~(\ref{eq:mom-heleigsrepresented2}) and Appendix~B the helicity measurement for $p=k$ corresponds to the trivial projection
\begin{equation}\label{eq:projpol}
    \Pi_{k}=\sum_{\s=\pm}\kbr{\s}{\s}_k=\begin{pmatrix}
                                                            1 & 0 & 0 \\
                                                            0 & 1 & 0 \\
                                                            0 & 0 & 0\\
                                                          \end{pmatrix},
    \qquad
    \ket{+}_k=\begin{pmatrix}
                               1 \\
                               0 \\
                                \,0\, \\
                             \end{pmatrix}
    \quad
    \ket{-}_k=\begin{pmatrix}
                               0 \\
                               1 \\
                               \,0\, \\
                             \end{pmatrix},
\end{equation}
where special symbols have been assigned for the helicity eigenstates because we work with them in the main body. The problem is that the projection is of rank two. When a single photon arrives in a direction $ p\not= k$ it  has a longitudinal component and hence it gets cut during the course of measurement. Then, to have a consistent definition of the helicity density matrix in a rotated reference frame, one has to reconstruct the longitudinal part because that is the place where some parts of the wave function `got lost'. This basically  means to measure the helicity in the planes perpendicular to the $x$ and $y$ axes in our fixed coordinate system to be able to build a three-dimensional helicity density matrix. Finally, we get a two-dimensional effective density matrix $\vr_{\rm{eff}}$ by cutting a $2\times2$ block from the whole three-dimensional matrix. The resulting state is positive-semidefinite and subnormalized. It is the consequence of the lemma  following the definition.
\begin{defi}
    For a matrix of dimension $n$ the leading principal submatrices  are all the upper left submatrices of dimension $k\leq n$.
\end{defi}
\begin{lem}\label{lem:submatrix}
    Let $\vr$ be a density matrix written in terms of (not necessarily orthogonal) pure states $\ket{\phi_i}$ such that $\vr=\sum_{i=1}^np_i\kbr{\phi_i}{\phi_i}$. Then every leading principal  submatrix of $\vr$ is positive semidefinite.
\end{lem}
\begin{proof}
    Taking the $k-$dimensional leading principal submatrix of $\ket{\phi_i}$ means to replace the last $n-k$ entries of  $\ket{\phi_i}$ by zeros. We will label the resulting vectors as $\ketf{\phi_i}$. The corresponding matrices $\kbf{\phi_i}{\phi_i}$ are still positive semidefinite and the positivity stays preserved by taking their convex combination $\tilde\vr=\sum_{i=1}^kp_i\kbf{\phi_i}{\phi_i}$. The resulting density matrix is subnormalized.
\end{proof}
\begin{col}
    The projection $\Pi_{ k}$  does exactly the job of cutting out the principal two-dimensional submatrix from the $3\times3$ helicity matrix (recall that the axial part is set to zero for all $ p$ by the gauge condition). Therefore
    \begin{equation}\label{eq:rhoeff}
            \vr_{\rm{eff}}={1\over N}\int{\rm d}\mu( p)|f( p)|^2\Pi_ k\left[(\a_ p\ket{+}_ p+\b_ p\ket{-}_ p)
            (\bar\a_ p\bra{+}_ p+\bar\b_ p\bra{-}_ p)\right]\Pi_ k,
    \end{equation}
    where $N$ is a normalization constant.
\end{col}
Eq.~(\ref{eq:rhoeff}) says nothing other than if we measure the helicity perpendicular to the propagation direction (the $z$ axis in this case and so $\Pi_ k$ is of the simple form Eq.~(\ref{eq:projpol})) a statistical mixture of all $x,y$ helicity components for all $ p$ is generated. This is exactly the  result from~\cite{pete03}. Of course, we may decide to measure along a general axis $ g\not= k$. In this case Eq.~(\ref{eq:rhoeff}) is still valid if we take $\Pi_ k\to\Pi_ g$ and render the helicity vectors in the new basis.
\begin{rem}
    One can adopt the approach of~\cite{AW1} where  the above Peres-Terno measurement scheme was reformulated in terms of measurement of the Stokes parameters of the incoming wave packet. We will not go into detail but we just point out the Lie algebraic aspect of this approach. In a similar manner as above the authors basically constructed three different effective helicity matrices which  correspond to the above mentioned measurement of the helicity in three orthogonal spatial directions. From these they constructed a three-dimensional helicity density matrix which in general can be written in terms of the $su(3)$ Lie algebra generators $\la_i^{su(3)}$~\cite{biedlouck}
    \begin{equation}\label{eq:su3lie}
        \vr={\openone\over3}+\sum_{i=1}^8s_i^{su(3)}\la_i^{su(3)},
    \end{equation}
    where $s_i^{su(3)}$ are the actual measured Stokes parameters. We now realize that the $su(3)$ Lie algebra is composed of three mutually dependent $su(2)$ Lie algebras and they precisely form the three effective helicity density matrices calculated in~\cite{AW2} (and one of them is Eq.~(\ref{eq:rhoeff})).
\end{rem}
\begin{rem}
    We should stress that in the above process one indeed gets a longitudinal part of the wave function but there is no need to call it a non-physical situation~\cite{pete03}. We simply get a helicity component in the plane not parallel to the plane where the helicity is measured. In principle, this situation cannot be avoided unless one works with single-particle states or very narrow photonic wave packets.

    Also, the three-dimensional helicity density matrix concept has only little to do with the bosonic nature of the photon. It is just a coincidence that we measure along three perpendicular directions corresponding to our usual spatial dimensions. If we had had a more sophisticated measurement  device we could have measured the helicity in each and every plane perpendicular to every $ p$ direction and map it to some other (non-photonic) multi-level quantum system. We would avoid projecting it down to our usual three-dimensional space where  helicity measurement devices usually operate. Pushing it to the limit, this way we would have been able to recover the whole density matrix because  $\brk{p,\pm}{p',\pm}=\d( p- p')$ holds in the infinite-dimensional helicity-momentum Hilbert space.
\end{rem}

So far we talked just how a wave packet transforms under an ordinary rotation. The remedy was to introduce a three-dimensional helicity density matrix or to eventually create a narrow wave packet where the `longitudinal' component is negligible~\cite{pete03}. The effect of a boost on a wave packet is more disastrous. Assume again that the wave packet Eq.~(\ref{eq:genwavepacket}) is propagating along the $z$ axis with the standard momentum $ k$ and an observer is also moving  along the $z$ axis with a constant velocity $v=v_z\in(-1,1)$. Having $\eta=\atanh{v_z}$ the boost $B_z: p\to q$ where $B_z(\eta)=\exp{(-i\eta K_3)}$ does not induce any phase change~\cite{ging} so the boost transforms a general single-particle state as $U(B_z(\eta)):\ketp\to\ket{B_zp,\s}$. The effect of the boost on a four-momentum vector is twofold: (i) The magnitude of the three-vector and the zeroth component are no longer equal to one. This is the Doppler effect we suppress by renormalization as discussed below Eq.~(\ref{eq:genwavepacket}). (ii) The gauge condition is violated since it is not a Lorentz covariant condition and must be imposed for every $ q$ as is usually done in quantum electrodynamics~\cite{Wein}. Thus $B_z:{\ve}_{ p,\s}\to {\ve}_{ q,\s}$ is followed by
\begin{equation}\label{eq:gaugecond}
    {\rm Coulomb\ gauge:\ } {\ve}_{ q,\s}\to {\ve}_{ q,\s}-g_{ q,\s} q,
\end{equation}
where $g_{ q,\s}={\ve^0_{ q,\s}/q^0}$ ($q^0=1$ following from (i)). The gauge condition Eq.~(\ref{eq:gaugecond}) assures that the new helicity four-vectors stay orthogonal in a plane orthogonal to $ q$. This is a final blow to our effort to have nice transformation properties of the helicity density matrix between two different frames of reference.
\begin{rem}
    The things get more technically complicated if we allow for a boost in a general direction. Then, a Wigner phase appears what can be easily seen if we realize that such a boost can be decomposed in a boost in the $z$ direction followed by a rotation in the required direction.
\end{rem}

\end{document}